\documentclass{article} 
\usepackage{iclr2022_conference,times}

\usepackage{amsmath,amsfonts,bm}

\def\eqref#1{equation~\ref{#1}}

\def\1{\bm{1}}

\DeclareMathAlphabet{\mathsfit}{\encodingdefault}{\sfdefault}{m}{sl}
\SetMathAlphabet{\mathsfit}{bold}{\encodingdefault}{\sfdefault}{bx}{n}

\usepackage{hyperref}
\usepackage{url}
\usepackage[utf8]{inputenc} 
\usepackage[T1]{fontenc}    
\usepackage{hyperref}       
\usepackage{url}            
\usepackage{booktabs}       
\usepackage{amsfonts}       
\usepackage{nicefrac}       
\usepackage{microtype}      
\usepackage{xcolor}         
\usepackage{amsmath}
\usepackage{amssymb}
\usepackage{algorithm}
\usepackage{algorithmic}
\usepackage{booktabs}
\usepackage{siunitx}
\usepackage{amsthm}
\usepackage{graphicx}

\usepackage{thmtools,thm-restate}

\title{Denoising Diffusion Gamma Models}

\author{
Eliya Nachmani* \\
Tel-Aviv University \\
Facebook AI Research \\
\texttt{enk100@gmail.com} \\
\And 
Robin San Roman* \\
École Normale Supérieure Paris-Saclay \\
\texttt{sanroman.robin@gmail.com} \\  
\And
Lior Wolf\\
Tel-Aviv University\\
\texttt{wolf@cs.tau.ac.il} \\  
}

\iclrfinalcopy 
\begin{document}

\maketitle
{\let\thefootnote\relax\footnote{{*Equal contribution}}}
\begin{abstract}
  Generative diffusion processes are an emerging and effective tool for image and speech generation. In the existing methods, the underlying noise distribution of the diffusion process is Gaussian noise. However, fitting distributions with more degrees of freedom could improve the performance of such generative models. In this work, we investigate other types of noise distribution for the diffusion process. Specifically, we introduce the Denoising Diffusion Gamma Model (DDGM) and show that noise from Gamma distribution provides improved results for image and speech generation. Our approach preserves the ability to efficiently sample state in the training diffusion process while using Gamma noise. 
\end{abstract}

\section{Introduction}
Deep generative neural networks have shown significant progress over the last years. The main architectures for generation are: (i) VAE \citep{kingma2013auto} based, for example, NVAE \citep{vahdat2020nvae} and VQ-VAE \citep{razavi2019generating}, (ii) GAN \citep{goodfellow2014generative} based, for example, StyleGAN \citep{karras2020analyzing} for vision application and WaveGAN \citep{donahue2018adversarial} for speech
(iii) Flow-based, for example Glow \citep{kingma2018glow} (iv) Autoregessive, for example, Wavenet for speech \citep{oord2016wavenet} and (v) Diffusion Probabilistic Models \citep{sohl2015deep}, for example, Denoising Diffusion Probabilistic Models (DDPM) \citep{ho2020denoising} and its implicit version DDIM \citep{song_denoising_2020}.

Models from this last family have shown significant progress in generation capabilities in the last years, e.g., \citep{chen_wavegrad_2020,kong_diffwave_2020}, and have achieved results comparable to state-of-the-art generation architecture for both images and speech. 

A DDPM is a Markov chain of latent variables. Two processes are modeled: (i) a diffusion process and (ii) a denoising process. During training, the diffusion process learns to transform data samples into Gaussian noise. Denoising is the reverse process and it is used during inference for generating data samples, starting from Gaussian noise. The second process can be conditioned on attributes to control the generation sample. To obtain high-quality synthesis, a large number of denoising steps is used (i.e. $1000$ steps). A notable property of the diffusion process is a closed-form formulation of the noise that arises from accumulating diffusion stems. This allows sampling arbitrary states in the Markov chain of the diffusion process without calculating the previous steps. 

In the Gaussian case, this property stems from the fact that adding Gaussian distributions leads to another Gaussian distribution. Other distributions have similar properties. For example, for the Gamma distribution, the sum of two distributions that share the scale parameter is a Gamma distribution of the same scale. The Poisson distribution has a similar property. However, its discrete nature makes it less suitable for DDPM. 

In DDPM, the mean of the distribution is set at zero. The Gamma distribution, with its two parameters (shape and scale), is better suited to fit the data than a Gaussian distribution with one degree of freedom (scale). Furthermore, the Gamma distribution generalizes other distributions, and many other distributions can be derived from it~\citep{leemis2008univariate}.

The added modeling capacity of the Gamma distribution can help speed up the convergence of the DDPM model. Consider, for example, a conventional DDPM model that was trained with Gaussian noise on the CelebA dataset~\citep{liu2015faceattributes}. 

The noise distribution throughout the diffusion process can be visualized by computing the histogram of the estimated residual noise in the generation process. The estimated residual noise $\hat{\epsilon}$ is given by $\hat{\epsilon} =\frac{\sqrt{\bar \alpha_t}x_0 - x_t}{\sqrt{1 - |\bar \alpha_t|}}$, where {$\bar \alpha_t$ is the noise schedule, $x_0$ is the data point and $x_t$ is the estimate state at timestep $t$}, as can be derived from Eq.4 from \citep{song_denoising_2020}. Both a Gaussian distribution and Gamma distribution can then be fitted to this histogram, as shown in Fig.~\ref{fig:fitting_error}(a,b). As can be seen, the Gamma distribution provides a better fit to the estimated residual noise $\hat{\epsilon}$. Moreover, Fig.~\ref{fig:fitting_error}(c) presents the mean fitting error between the histogram and the fitted probability distribution function. Evidently, the Gamma distribution is a better fit than the Gaussian distribution.

In this paper, we investigate the non-Gaussian Gamma noise distribution. The proposed models maintain the property of the diffusion process of sampling arbitrary states without calculating the previous steps. Our results are demonstrated in two major domains: vision and audio. In the first domain, the proposed method is shown to provide a better FID score for generated images. For speech data, we show that the proposed method improves various measures, such as Perceptual Evaluation of Speech Quality (PESQ) and short-time objective intelligibility (STOI).

\begin{figure}[]
    \centering
    \begin{tabular}{cccc}
    \includegraphics[width=.3\textwidth,keepaspectratio]{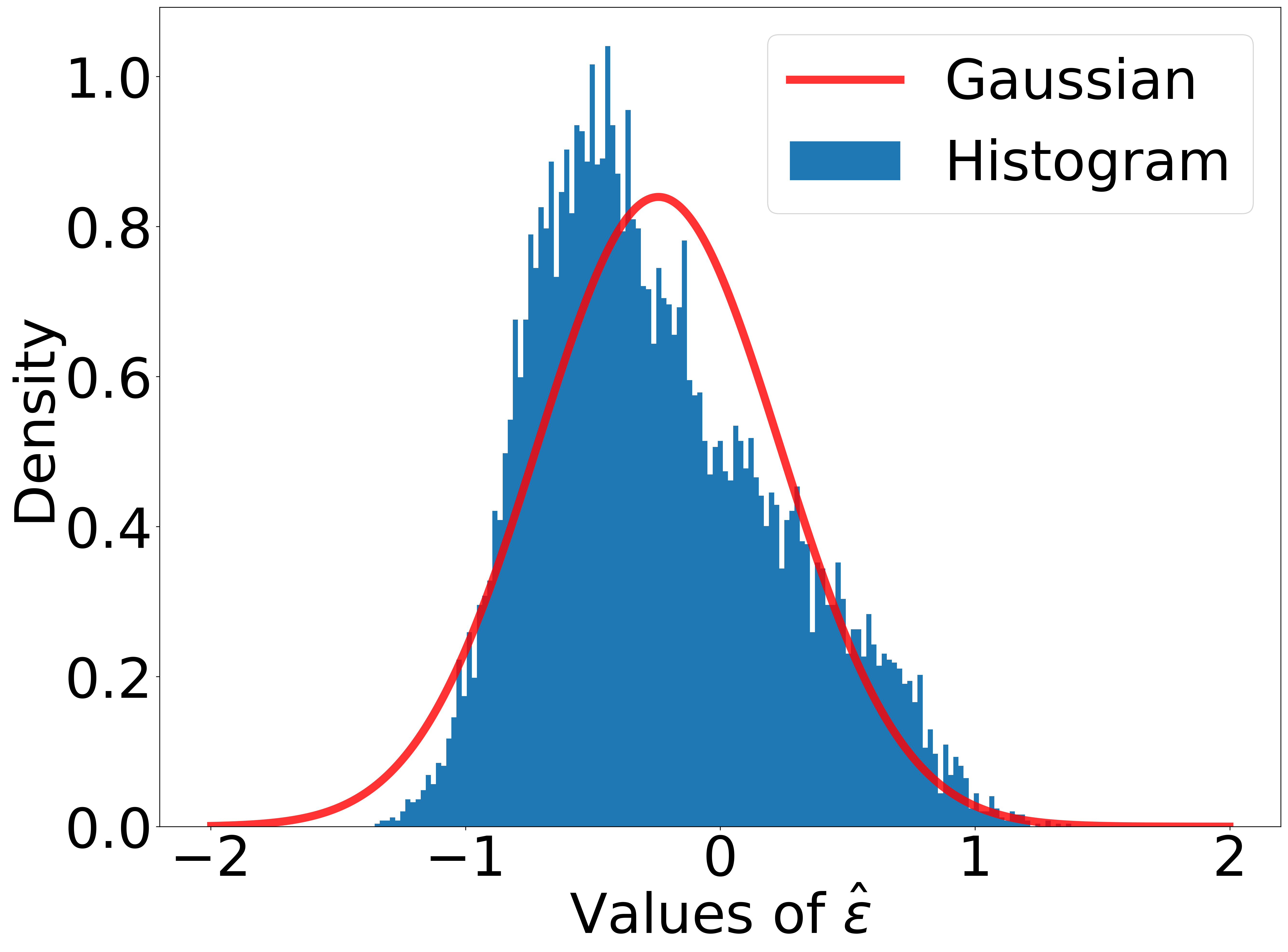} &
    \includegraphics[width=.3\textwidth,keepaspectratio]{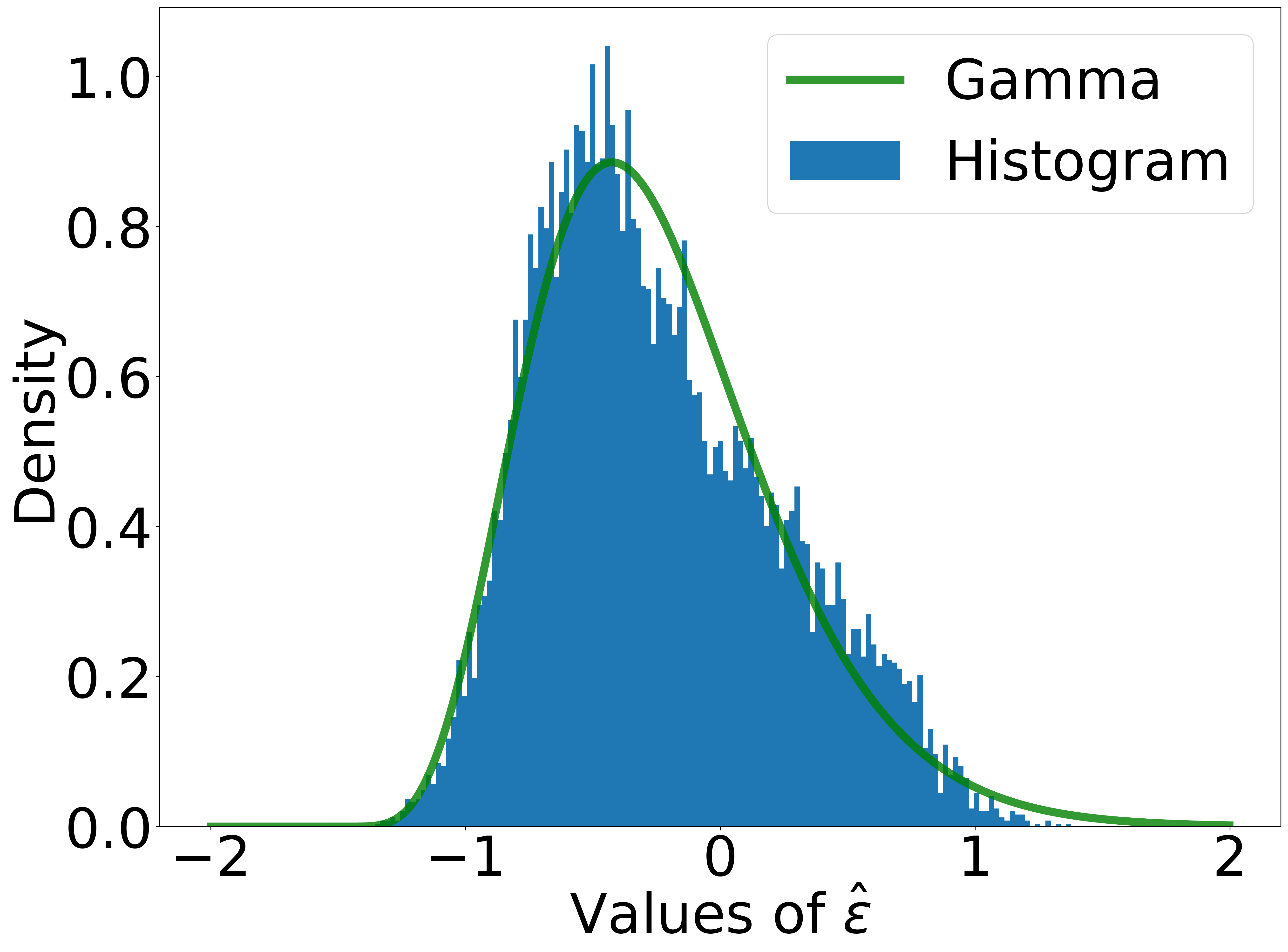} &
    \includegraphics[width=.3\textwidth,keepaspectratio]{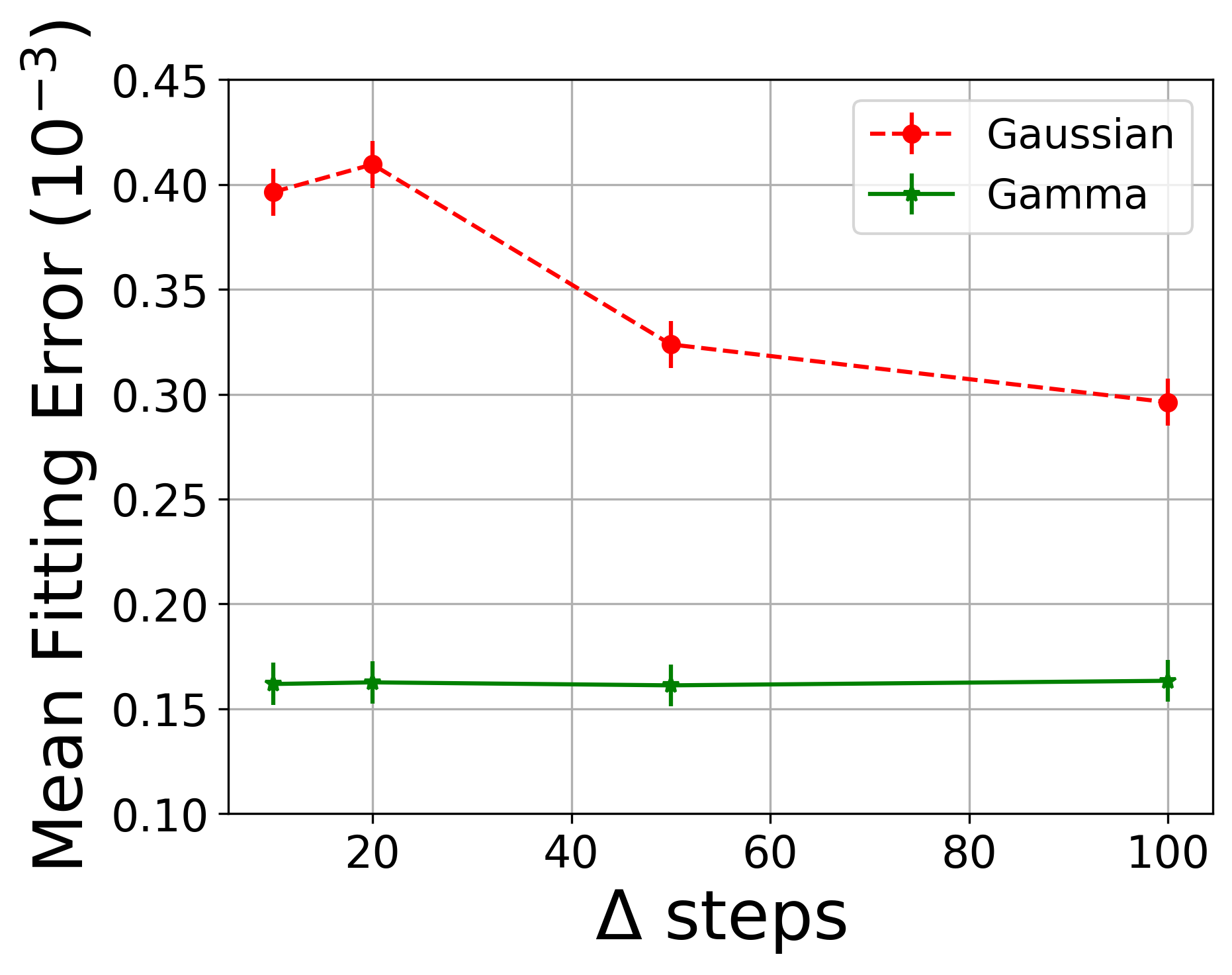} \\
    (a) & (b) & (c) \\
    \end{tabular}
    \caption{{Fitting a distribution to the histogram of the generation error, which given by the scaled difference between $x_0$ and the image $x_t$ after $t$ DDPM steps $\hat{\epsilon}=\frac{\sqrt{\bar \alpha_t}x_0 - x_t}{\sqrt{1 - |\bar \alpha_t|}}$. The model is a pretrained DDPM (Gaussian) celebA (64x64) model.} (a) The fitting of a Gaussian to the histogram of a typical image after $t-50$ steps. (b) Fitting a Gamma distribution. (c) The fitting error to Gaussian and Gamma distribution, measured as the MSE between the histogram and the fitted probability distribution function. Each point is the average value for the generation of $100$ images. The vertical error bars denote the standard deviation.}
    \label{fig:fitting_error}
\end{figure}

\section{Related Work} 
In their seminal work, \cite{sohl2015deep} introduce the Diffusion Probabilistic Model. This model is applied to various domains, such as time series and images. The main drawback in the proposed model is that it needs up to thousands of iterative steps to generate a valid data sample. \cite{song2019generative} proposed a diffusion generative model based on Langevin dynamics and the score matching method~\citep{hyvarinen2005estimation}. The model estimates the Stein score function~\citep{liu2016kernelized} which is the logarithm of data density. Given the Stein score function, the model can generate data points.

Denoising Diffusion Probabilistic Models (DDPM)~\citep{ho2020denoising} combine generative models based on score matching and neural Diffusion Probabilistic Models into a single model. Similarly, in \cite{chen_wavegrad_2020,kong2020diffwave} a generative neural diffusion process based on score matching was applied to speech generation. These models achieve state-of-the-art results for speech generation, and show superior results over well-established methods, such as Wavernn \citep{kalchbrenner2018efficient}, Wavenet \citep{oord2016wavenet}, and GAN-TTS \citep{binkowski2019high}.

Diffusion Implicit Models (DDIM) offer a way to accelerate the denoising process~\citep{song_denoising_2020}. The model employs a non-Markovian diffusion process to generate a higher quality sample. The model helps reduce the number of diffusion steps, e.g., from a thousand steps to a few hundred.

{\cite{dhariwal2021diffusion} find a better diffusion architecture through a series of exploratory experiments, leading to the Ablated Diffusion Model (ADM). This model outperforms the state-of-the-art in image synthesis, which was previously provided by GAN based-models, such as BigGAN-deep \citep{brock2018large} and StyleGAN2 \citep{karras2020analyzing}. 
ADM is further improved using a novel Cascaded Diffusion Model (CDM). Our contribution is fundamental and can be incorporated into the proposed ADM and CDM architectures.

\cite{watson2021learning} proposed an efficient method for sampling from diffusion probabilistic models by a dynamic programming algorithm that finds the optimal discrete time schedules. \cite{choi2021ilvr} introduces the Iterative Latent Variable Refinement (ILVR) method for guiding the generative process in DDPM. Moreover, \cite{kong2021fast} systematically investigates fast sampling methods for diffusion denoising models. \cite{lam2021bilateral} propose bilateral denoising diffusion models (BDDM), which take significantly fewer steps to generate high-quality samples.

\cite{huang2021variational} derive a variational framework for likelihood estimating of the marginal likelihood of
continuous-time diffusion models. Moreover, \cite{kingma2021variational} shows equivalence between various diffusion processes by using a simplification of the variational lower bound (VLB).
}

\cite{song2020score} show that score-based generative models can be considered a solution to a stochastic differential equation. \cite{gao2020learning} provide an alternative approach for training an energy-based generative model using a diffusion process.

Another line of work in audio is that of neural vocoders based on a denoising diffusion process. WaveGrad~\citep{chen_wavegrad_2020} and DiffWave~\citep{kong2020diffwave} are conditioned on the mel-spectrogram and produce high-fidelity audio samples, using as few as six steps of the diffusion process. These models outperform adversarial non-autoregressive baselines. \cite{popov2021grad} propose a text-to-speech diffusion base model, which allows generating speech with the flexibility of controlling the trade-off between sound quality and inference speed.

Diffusion models were also applied to natural language processing tasks. \cite{hoogeboom2021argmax} proposed a multinomial diffusion process for categorical data and applied it to language modeling. \cite{austin2021structured} generalize the multinomial diffusion process with Discrete Denoising Diffusion Probabilistic Models (D3PMs) and improve the generated results for the text8 and One Billion Word (LM1B) datasets.

\begin{figure}[!t]
\begin{minipage}[t]{0.49\textwidth}
\begin{algorithm}[H]
  \caption{DDPM training procedure.} 
  \label{alg:DDPM_train}
  \begin{algorithmic}[1]
    \STATE Input: dataset $d$, diffusion process length $T$, noise schedule $\beta_1,...,\beta_T$
   \REPEAT
   \STATE $ x_0 \sim d(x_0)$
   \STATE $t \sim \mathcal{U}(\{1, ..., T\})$
   \STATE $\varepsilon \sim \mathcal{N}(0, I)$
   \STATE $x_t = \sqrt{\bar \alpha_t}x_0 + \sqrt{1 - \bar \alpha_t}\varepsilon$
   \STATE Take gradient descent step on: \newline $\| \varepsilon - \varepsilon_{\theta}(x_t, t )\|_1 $
   \UNTIL{converged}
  \end{algorithmic}
\end{algorithm}
\end{minipage}
\hfill
\begin{minipage}[t]{0.49\textwidth}
\begin{algorithm}[H]
  \caption{DDPM sampling algorithm}%
  \label{alg:DDPM_inf}
  \begin{algorithmic}[1]
  \STATE $x_T \sim \mathcal{N}(0, I) $
   \FOR{t= T, ..., 1}
   \STATE $ z \sim \mathcal{N}(0, I)$
   \STATE $\hat \varepsilon = \varepsilon_\theta(x_t, t)$
   \STATE $x_{t-1} = \frac{x_t - \frac{1 - \alpha_t}{\sqrt{1 - \bar \alpha_t}}\hat \varepsilon }{\sqrt{\alpha_t}}$ 
   \IF{$t \neq 1$}
   \STATE $x_{t-1} = x_{t-1} + \sigma_t z$ 
   \ENDIF
   \ENDFOR
   \STATE \textbf{return} $x_0$
  \end{algorithmic}
\end{algorithm}
\end{minipage}
\end{figure}

\section{Diffusion models for Gamma Distribution}\label{section:method} 

We start by recapitulating the Gaussian case, after which we derive diffusion models for the Gamma distribution.

\subsection{Background - Gaussian DDPM} 
Diffusion networks learn the gradients of the data log density: 
\begin{equation}
    s(y) = \nabla_y \log p(y) 
\end{equation}
By using Langevin Dynamics and the gradients of the data log density $\nabla_y \log p(y)$, a sample procedure from the probability can be done by:
\begin{equation}\label{eq:langevin}
    \tilde y_{i+1} = \tilde y_i + \frac{\eta}{2} s(\tilde y_i) + \sqrt{\eta}z_i
\end{equation}
where $z_i\sim \mathcal{N}(0, I)$ and $\eta > 0$ is the step size.

The diffusion process in DDPM \citep{ho2020denoising} is defined by a Markov chain that gradually adds Gaussian noise to the data according to a noise schedule. The diffusion process is defined by:
\begin{equation}
    q(x_{1:T}|x_0) = \prod_{t=1}^{T} q(x_t|x_{t-1})\,,
\end{equation}
where T is the length of the diffusion process, and $x_T,...,x_t,x_{t-1},...,x_0$ is a sequence of latent variables with the same size as the clean sample $x_0$. 

The Diffusion process is parameterized with a set of parameters called noise schedule ($\beta_1, \dots \beta_T$), which defines the variance of the noise added at each step:
\begin{equation}
\label{eq:ddpm_diff_normal}
     q(x_t|x_{t-1}) := \mathcal{N}(x_{t}; \sqrt{1 - \beta_t}x_{t-1}, \beta_t\mathbf{I})\,,
\end{equation}

Since we are using a Gaussian noise random variable at each step, the diffusion process can be simulated for any number of steps with the closed formula: 
\begin{equation}\label{eq:closedy_n}
    x_t = \sqrt{\bar \alpha_t} x_0 + \sqrt{1 - \bar\alpha_t} \varepsilon\,,
\end{equation}

where $\alpha_i = 1 - \beta_i$, $\bar \alpha_t = \prod_{i=1}^t \alpha_i$ and $\varepsilon = \mathcal{N}(0,\mathbf{I})$.

Diffusion models are a class of generative neural network of the form $p_\theta(x_0)=\int p\theta(x_{0:T}) dx_{0:T}$ that learn to reverse the diffusion process. One can write that:

\begin{equation}
    p_\theta(x_0) = p(x_T)\prod_{t=1}^{T}p_\theta(x_{t-1}|x_t)
\end{equation}

As described in \citep{ho2020denoising}, one can learn to predict the noise present in the data with a network $\varepsilon_\theta$ and sample from $p_\theta(x_{t-1}|x_t)$ using the following formula :

\begin{equation}\label{eq:DDPM_update}
x_{t-1} = \dfrac{x_t - \frac{1 - \alpha_t}{\sqrt{1- \bar\alpha_t}}\varepsilon_\theta(x_t, t)}{\sqrt{\bar\alpha_t}} + \sigma_t \varepsilon\,,
\end{equation}

where $\varepsilon$ is white noise and $\sigma_t$ is the standard deviation of added noise. \citep{song_denoising_2020} use ${\sigma_t}^2 = \beta_t$.

The training procedure of $\varepsilon_{\theta}$ is defined in Alg.\ref{alg:DDPM_train}. Given the input dataset $d$, the algorithm samples $\epsilon$, $x_0$ and $t$. The noisy latent state $x_t$ is calculated and fed to the DDPM neural network $\varepsilon_\theta$. A gradient descent step is taken in order to estimate the $\varepsilon$ noise with the DDPM network $\varepsilon_\theta$.

The complete inference algorithm present at Alg.~\ref{alg:DDPM_inf}. Starting from Gaussian noise and then reversing the diffusion process step-by-step, by iteratively employing the update rule of Eq.~\ref{eq:DDPM_update}.


\subsection{Denoising Diffusion Gamma Models (DDGM)} 
We expand the framework of diffusion generative processes by incorporating a new noise distribution, namely the Gamma Distribution. We call this new type of models Denoising Diffusion Gamma Models. First, we define the Gamma diffusion process, then we present a way to sample from this process, and finally we show how to train those models by computing the variational lower bound and deriving a novel loss function from it.

\subsubsection{The Gamma Model}
In the Gaussian case the diffusion equation (Eq.~\ref{eq:ddpm_diff_normal}) can be written as:
\begin{equation}\label{eq:difusion_step}
	 x_t = \sqrt{1 - \beta_t} x_{t-1} + \sqrt{\beta_t}\epsilon_t
\end{equation}
where $\epsilon_t$ is the Gaussian noise of step $t$.
One can denote $\Gamma(k, \theta)$ as the Gamma distribution, where $k$ and $\theta$ are the shape and the scale respectively. We modify Eq.~\ref{eq:difusion_step} by adding, during the diffusion process, noise that follows a Gamma distribution:
\begin{equation}\label{eq:gamma_t}
	 x_t = \sqrt{1 - \beta_t} x_{t-1} + (g_t - \mathbb{E}(g_t))
\end{equation}
where $g_t\sim \Gamma(k_t, \theta_t)$, $\theta_t = \sqrt{\bar \alpha_t}\theta_0$ and $k_t=\dfrac{\beta_t}{\alpha_t{\theta_0}^2}$. Note that $\theta_0$ and $\beta_t$ are hyperparameters.

Since the sum of Gamma distribution (with the same scale parameter) is distributed as Gamma distribution, one can derive a closed form for $x_t$, i.e. an equation to calculate $x_t$ from $x_0$:
\begin{equation}
\label{eq:close_form_single_gamma}
 x_t = \sqrt{\bar \alpha_t} x_0 + (\bar g_t - \bar k_t\theta_t)
\end{equation}
where $\bar g_t \sim \Gamma(\bar k_t, \theta_t)$ and $\bar k_t = \sum_{i=1}^t k_i$. 

\begin{restatable}{lem}{primelemma}
\label{lm:lm1}
    Let $\theta_0 \in \mathbb{R}$, 
    Assuming $\forall t \in \{1,..., T\}$, $k_t=\dfrac{\beta_t}{\alpha_t{\theta_0}^2}$, 
    $\theta_t = \sqrt{\bar \alpha_t}\theta_0$, and $g_t\sim \Gamma(k_t, \theta_t)$. Then $\forall t \in \{1,..., T\}$ the following hold:
    
    \begin{equation}\label{eq:lemma2_1}
        E(g_t - E(g_t)) =0, V(g_t - E(g_t)) =\beta_t
    \end{equation}
    
    \begin{equation}\label{eq:lemma2_2}
        x_t = \sqrt{\bar \alpha_t}x_0 + (\bar g_t - E(\bar g_t))
    \end{equation}
    where $\bar g_t \sim \Gamma(\bar k_t, \theta_t)$ and $\bar k_t = \sum_{i=1}^t k_i$
\end{restatable}

The complete proof for Lemma \ref{lm:lm1} is given in Appendix~\ref{app:sec_1}.

Similarly to Eq.\ref{eq:DDPM_update} by using Langevin dynamics, the inference is given by:
\begin{equation}\label{eq:sample_gamma}
x_{t-1} = \dfrac{x_t - \frac{1 - \alpha_t}{\sqrt{1- \bar\alpha_t}}\varepsilon_\theta(x_t, t)}{\sqrt{\bar\alpha_t}} + \sigma_t \dfrac{\bar g_t - E(\bar g_t)}{\sqrt{V(\bar g_t)}}
\end{equation}

In Algorithm~\ref{alg:single_gamma_train} we describe the training procedure. As input we have the: (i) initial scale $\theta_0$, (ii) the dataset $d$, (iii) the maximum number of steps in the diffusion process $T$ and (iv) the noise schedule $\beta_1,...,\beta_T$. The training algorithm sample: (i) an example $x_0$, (ii) number of step $t$ and (iii) noise $\varepsilon$. Then it calculates $x_t$ from $x_0$ by using Eq.\ref{eq:close_form_single_gamma}. The neural network $\varepsilon_{\theta}$ has an input $x_t$ and is conditional on the time step $t$.
Next, it takes a gradient descent step to approximate the normalized noise $\frac{\bar g_t - \bar k_t \theta_t}{\sqrt{1 - | \bar \alpha_t|}}$ with the neural network $\varepsilon_{\theta}$. The main changes between Algorithm~\ref{alg:single_gamma_train} and the single Gaussian case (i.e. Alg.~\ref{alg:DDPM_train}) are the following: (i) calculating the Gamma parameters, (ii) $x_t$ update equation and (iii) the gradient update equation. 

The inference procedure is given in Algorithm~\ref{alg:single_gamma_infernce}. It starts from a zero mean noise $x_T$ sampled from $\Gamma(\theta_T, \bar k_T)$. Next, for $T$ steps the algorithm estimates $x_{t-1}$ from $x_t$ by using Eq.\ref{eq:sample_gamma}. Note that as in \citep{song_denoising_2020} $\sigma_t=\beta_t$. Algorithm~\ref{alg:single_gamma_infernce} replaces the  Gaussian version (i.e. Alg.~\ref{alg:DDPM_inf}) with the following: (i) the starting sampling point $x_T$, (ii) the sampling noise $z$ and (iii) the $x_t$ update equation.

\begin{figure}[!t]
\begin{minipage}[t]{0.49\textwidth}
\begin{algorithm}[H]
  \caption{Gamma Training Algorithm} 
  \label{alg:single_gamma_train}
  \begin{algorithmic}[1]
  \STATE Input: initial scale $\theta_0$, dataset $d$, diffusion process length $T$, noise schedule $\beta_1,...,\beta_T$
   \REPEAT
   \STATE $x_0 \sim d(x_0)$
   \STATE $t \sim \mathcal{U}(\{1, ..., T\})$
   \STATE $\bar g_t \sim \Gamma(\bar k_t, \theta_t)$
   \STATE $x_t = \sqrt{\bar \alpha_t}x_0 + (\bar g_t -\bar k_t\theta_t)$
   \STATE Take a gradient descent step on: \newline $\left | \frac{\bar g_t - \bar k_t \theta_t}{\sqrt{1 - | \bar \alpha_t|}} - \varepsilon_{\theta}(x_t, t)\right| $
   \UNTIL{converged}
  \end{algorithmic}
\end{algorithm}
\end{minipage}
\hfill
\begin{minipage}[t]{0.49\textwidth}
\begin{algorithm}[H]
  \caption{Gamma Inference Algorithm}%
  \label{alg:single_gamma_infernce}
  \begin{algorithmic}[1]
   \STATE $ \gamma \sim \Gamma(\theta_T, \bar k_T) $
   \STATE $x_T = \gamma - \theta_T*\bar k_T$
   \FOR{t = T, ..., 1}
      \STATE $x_{t-1} = \frac{x_t -\frac{1-\alpha_t}{\sqrt{1 - \bar\alpha_t}} \epsilon(x_t, t)}{\sqrt{\alpha_t}}$
    \IF{t > 1}
        \STATE $z \sim \Gamma(\theta_{t-1}, \bar k_{t-1})$
        \STATE $z = \frac{z - \theta_{t-1}\bar k_{t-1}}{\sqrt{(1- \bar \alpha_t)}}$
        \STATE $x_{t-1} = x_{t-1} + \sigma_t z$
    \ENDIF
    \ENDFOR
  \end{algorithmic}
\end{algorithm}
\end{minipage}
\end{figure}

{\subsubsection{The Reverse Process for DDGM}}
The reverse process $q(x_{t-1}|x_0,x_t)$ defines the underlying generation process. Therefore, in this section, we will obtain the reverse process for the Gamma denoising diffusion model. Furthermore, we will use the reverse process $q(x_{t-1}|x_0,x_t)$ to obtain the variational lower bound and the appropriate loss function for the Gamma distribution denoising diffusion model.

The reverse process is given by:
\begin{equation}
\label{eq:reverse_process}
    q(x_{t-1}| x_0, x_t) = q(x_t| x_{t-1}, x_0)\dfrac{q(x_{t-1}|x_0)}{q(x_t|x_0)} 
\end{equation}

Next, one can calculate each one of the three main components of the reverse process, i.e. (i) $q(x_t| x_{t-1}, x_0)$, (ii) $q(x_{t-1}|x_0)$ and (iii) $q(x_t|x_0)$.

Since $q$ is memoryless, $q(x_t| x_{t-1}, x_0)= q(x_t| x_{t-1})$. Therefore, the first component (i) of Eq.~\ref{eq:reverse_process} is the forward process. The forward process is given by: 
\begin{align}
    q(x_t|x_{t-1}) &= p(g_t = x_t - \sqrt{1-\beta_t}x_{t-1}+ k_t\theta_t) \\ 
    &= \dfrac{(x_t - \sqrt{1-\beta_t}x_{t-1}+ k_t\theta_t)^{k_t-1} e^{-(x_t - \sqrt{1-\beta_t}x_{t-1}+ k_t\theta_t)/\theta_t}}{\Gamma(k_t)\theta_t}
\end{align}

The second component of Eq.\ref{eq:reverse_process} is given by:
\begin{equation}
q(x_{t-1}|x_0)= \dfrac{(x_{t-1} - \sqrt{\bar\alpha_{t-1}}x_0 +\bar k_{t-1}\theta_{t-1})^{\bar k_{t-1}-1} e^{-(x_t - \sqrt{\bar\alpha_{t-1}}x_0 +\bar k_{t-1}\theta_{t-1})/\theta_{t-1}}}{\Gamma(\bar k_{t-1}) \theta_t^{\bar k_{t-1}}}
\end{equation}

Similarly, the third component of Eq.\ref{eq:reverse_process} is given by:
\begin{equation}
q(x_t|x_0) = p( \bar g_t = x_t - \sqrt{\bar\alpha_t}x_0 +\bar k_t\theta_t) = \dfrac{(x_t - \sqrt{\bar\alpha_t}x_0 +\bar k_t\theta_t)^{\bar k_t-1} e^{-(x_t - \sqrt{\bar\alpha_t}x_0 +\bar k_t\theta_t)/\theta_t}}{\Gamma(\bar k_t) \theta_t^{\bar k_t}}
\end{equation}

Overall, the reverse process $q(x_{t-1}| x_0, x_t)$ is given by:
\begin{align}
\begin{aligned}
q(x_{t-1}| x_0, x_t) &= \dfrac{\left( (x_t - \sqrt{1-\beta_t}x_{t-1}+ k_t\theta_t)^{k_t-1} e^{-(x_t - \sqrt{1-\beta_t}x_{t-1}+ k_t\theta_t)/\theta_t}\right)}{\Gamma(k_t)\theta_t} \\
& \cdot \dfrac{\left( (x_{t-1} - \sqrt{\bar\alpha_{t-1}}x_0 +\bar k_{t-1}\theta_{t-1})^{\bar k_{t-1}-1} e^{-(x_t - \sqrt{\bar\alpha_{t-1}}x_0 +\bar k_{t-1}\theta_{t-1})/\theta_{t-1}} \right)}{\Gamma(\bar k_{t-1}) \theta_t^{\bar k_{t-1}}} \\
& \cdot \dfrac{\Gamma(\bar k_t) \theta_t^{\bar k_t}}{\left( (x_t - \sqrt{\bar\alpha_t}x_0 +\bar k_t\theta_t)^{\bar k_t-1} e^{-(x_t - \sqrt{\bar\alpha_t}x_0 +\bar k_t\theta_t)/\theta_t} \right)}
\end{aligned}
\end{align}

One can denote:
\begin{enumerate}
    \item $X_t = x_t - \sqrt{1-\beta_t}x_{t-1}+ k_t\theta_t$
    \item $\bar X_t = x_t - \sqrt{\bar\alpha_t}x_0 +\bar k_t\theta_t$ 
    \item $\bar X_{t-1} = x_{t-1} - \sqrt{\bar\alpha_{t-1}}x_0 +\bar k_{t-1}\theta_{t-1}$
\end{enumerate}

Thus, the reverse process $q(x_{t-1}| x_0, x_t)$ is proportional to:
\begin{equation}
\label{eq:reverse_prop}
q(x_{t-1}| x_0, x_t) \propto \dfrac{X_t^{k_t-1} e^{-X_t/\theta_t} \bar X_{t-1}^{\bar k_{t-1}-1} e^{-\bar X_{t-1}/\theta_{t-1}}}{\bar X_t^{\bar k_t-1} e^{-\bar X_t/\theta_t}} 
\end{equation}

{\subsubsection{Variational Lower Bound for DDGM}} 
\label{sec:vlb}

Denoising diffusion models \citep{ho2020denoising} trained by optimizing the usual variational bound on negative log likelihood:
\begin{equation}
E\left [ -log(p_{\theta}(x_0) \right ] \leq E_q\bigg[ -\log p(x_T) - \sum_{t \geq 1} \log \frac{p_\theta(x_{t-1} | x_t)}{q(x_t|x_{t-1})} \bigg] = L_{VLB}    
\end{equation}

To get the variational lower bound for the proposed Gamma denoising diffusion model, one can use Eq.5 from \cite{ho2020denoising}:
\begin{equation}
L_{VLB} = E_{q} \left [ L_T + \sum_{t>1}L_{t-1} + L_0  \right ]
\end{equation}
where $L_T,L_{t-1}$ and $L_0$ define by:
\begin{enumerate}
\item $L_T = D_{KL}(q(x_{T}|x_0)||q(x_T))$ 
\item $L_{t-1} = D_{KL}(q(x_{t-1}| x_0, x_t)||q(x_{t-1} | \hat x_0, x_t))$
\item $L_0=-\log(q(x_0|x_1))$ 
\end{enumerate}

$L_T$ is constant and ignored during training since it doesn't have learnable parameters. Moreover, in \citep{ho2020denoising} $L_0$ modeled with discrete decoder, however, in our proposed model we empirically found that the impact $L_0$ is negligible and can be removed.

Therefore, to calculate the variatonal lower bound one needs to obtain:
\begin{equation}
L_{t-1} = D_{KL}(q(x_{t-1}| x_0, x_t)||q(x_{t-1} | \hat x_0, x_t))
\end{equation}

where:
\begin{equation}
\label{eq:hatx_0}
\hat x_0 = \dfrac{x_t -\sqrt{1 - \bar \alpha_t} \varepsilon_\theta(x_t, t)}{\sqrt{\bar \alpha_t}}
\end{equation}

We can calculate the KL divergence with the exact form:
\begin{equation}
\label{eq:kl_large}
D_{KL}(q(x_{t-1}| x_0, x_t)||q(x_{t-1} | \hat x_0, x_t)) = E_{q(x_{t-1}| x_0, x_t)} \log\left(\dfrac{q(x_{t-1}| x_0, x_t)}{q(x_{t-1}| \hat x_0, x_t)}\right)
\end{equation}

Using Eq.\ref{eq:reverse_prop} the RHS of Eq.\ref{eq:kl_large} become: 
\begin{equation}
\log \left(\dfrac{q(x_{t-1}| x_0, x_t)}{q(x_{t-1}| \hat x_0, x_t)}\right) =  (\bar k_{t-1}-1) \log(\frac{\bar X_{t-1}}{\hat X_{t-1}}) - \frac{\bar X_{t-1} - \hat X_{t-1}}{\theta_{t-1}} -(\bar k_t -1)\log(\frac{\bar X_t}{\hat X_t}) +\frac{\bar X_t - \hat X_t}{\theta_t}
\end{equation}
One can show that the four terms present in the previous equation can be upper bounded with the L1 distance between the predicted $\hat x_0$ and the ground truth $x_0$:
\begin{itemize}

\item $|\frac{\bar X_{t-1} - \hat X_{t-1}}{\theta_{t-1}}| = |(x_0 - \hat x_0) \frac{\sqrt{\bar \alpha_{t-1}}}{\theta_{t-1}}| \leq C_1 |x_0 - \hat x_0| $

\item $|\frac{\bar X_t - \hat X_t}{\theta_{t}}| = |(x_0 - \hat x_0) \frac{\sqrt{\bar \alpha_t}}{\theta_t}| \leq C_2 |x_0 - \hat x_0|$

\item $(\bar k_t-1) \log(\frac{\bar X_t}{\hat X_t}) = (\bar k_t-1)\log\left(\dfrac{x_t - \sqrt{\bar\alpha_t}x_0 +\bar k_t\theta_t}{x_t - \sqrt{\bar\alpha_t}\hat x_0 +\bar k_t\theta_t}\right) = \log \left( 1 + \dfrac{\sqrt{\bar \alpha_t}(x_0 - \hat x_0)}{x_t - \sqrt{\bar\alpha_t}\hat x_0 +\bar k_t\theta_t} \right) \leq |\dfrac{\sqrt{\bar \alpha_t}(x_0 - \hat x_0)}{x_t - \sqrt{\bar\alpha_t}\hat x_0 +\bar k_t\theta_t}| = \frac{C_3}{\bar g_t} |x_0 - \hat x_0|$

\item 
$(\bar k_{t-1}-1) \log(\frac{\bar X_{t-1}}{\hat X_{t-1}}) = \log \left( 1 + \dfrac{\sqrt{\bar \alpha_{t-1}}(x_0 - \hat x_0)}{x_{t-1} - \sqrt{\bar\alpha_{t-1}}\hat x_0 +\bar k_{t-1}\theta_{t-1}} \right)  \\ \leq |\dfrac{\sqrt{\bar \alpha_{t-1}}(x_0 - \hat x_0)}{x_{t-1} - \sqrt{\bar\alpha_{t-1}}\hat x_0 +\bar k_{t-1}\theta_{t-1}}| = \frac{C_4}{\bar g_{t-1}} |x_0 - \hat x_0|$

\end{itemize}

The complete form of the $L_{t-1}$ upper bound can be expressed as follows:
\begin{equation}
\label{eq:vlb_end}
L_{t-1} \leq E_{q(x_{t-1}| x_0, x_t)} \left( C_1 +C_2 + \frac{C_3}{\bar g_t} + \frac{C_4}{\bar g_{t-1}}\right) |x_0 - \hat x_0| = \left( C_1 +C_2 + \frac{C_3}{\bar g_t} + \frac{C_4}{\bar g_{t-1}}\right) |x_0 - \hat x_0|
\end{equation}

As can be seen, the variational lower bound is bounded by some constant forms multiplied by the L1 norm between the data point $x_0$ and its estimation $\hat x_0$. The constant terms $C_1,C_2,C_3$ and $C_4$ as well as $\bar g_t$ and $\bar g_{t-1}$ are known values during the training. 

{\subsubsection{Loss Function for DDGM}}

Denoising diffusion probabilistic models use the variational lower bound to minimize the negative log likelihood. As described in Sec.\ref{sec:vlb}, one can minimize the variational lower bound by $L_t$ for $t \geq 1$. To do so, one can minimize the L1 norm from Eq.\ref{eq:vlb_end}. Our model optimizes the L1 norm between the sampled noise $\epsilon_\theta$ and the estimated noise $\varepsilon_\theta$. This is verified in the following lemmas.
\begin{restatable}{lem}{primelemmasec}
\label{lm:lm2}
    Minimizing the variational lower bound for DDGM (i.e. $L_t$ for $t \geq 1$) is equivalent to minimizing the L1 norm between the sampled noise and the estimated noise:
    \begin{equation}
        \mathcal{L} = \left |\dfrac{\bar g_t - \bar k_t \theta_t}{\sqrt{1 - \bar \alpha_t}} - \varepsilon_\theta(x_t, t)\right|
    \end{equation}
\end{restatable}
The complete proof for Lemma \ref{lm:lm2} is given in Sec.\ref{app:sec_2} at the appendix. Thus, the loss that is used in the Alg.\ref{alg:single_gamma_train} is given by $\mathcal{L} =\left | \dfrac{\bar g_t - \bar k_t \theta_t}{\sqrt{1 - \bar \alpha_t}} - \varepsilon_\theta(x_t, t)\right|$.

\section{Experiments} \label{sec:Exp}

\subsection{Speech Generation} 
For our speech experiments we used a version of Wavegrad \citep{chen_wavegrad_2020} based on this implementation \cite{ivangit} (under BSD-3-Clause License). We evaluate our model with high-level perceptual quality of speech measurements, PESQ \citep{PESQ_paper} and STOI \citep{STOI_paper}. We used the standard Wavegrad method with the Gaussian diffusion process as a baseline. We use two Nvidia Volta V100 GPUs to train our models.

For all the experiments, the inference noise schedules ($\beta_0, .., \beta_T$) were defined as described in the Wavegrad paper \citep{chen_wavegrad_2020}. For $1000$ and $100$ iterations the noise schedule is linear, for $25$ iterations it comes from the Fibonacci and for $6$ iterations we performed a model-dependent grid search to find the best noise schedule parameters. For other hyper-parameters (e.g. learning rate, batch size, etc) we use the same as in Wavegrad~\citep{chen_wavegrad_2020}. Training was performed using the following form of Eq.~\ref{eq:gamma_t}, e.g. $\theta_t = \sqrt{\bar \alpha_t} \theta_0$ and $k_t = \dfrac{\beta_t}{\bar\alpha_t{\theta_0}^2}$. Our best results were obtained using $\theta_0 = 0.001$. 

\noindent{\bf Results\quad }
Tab.~\ref{tab:LJ} presents the PESQ and STOI measurement for the LJ dataset \citep{ljspeech17}. As can be seen, for the proposed Gamma denoising diffusion model our results are better than the Wavegrad baseline for all number of iterations in both PESQ and STOI.

\begin{table}[]
\centering
\caption{PESQ and STOI metrics for the LJ dataset for various Wavegrad-like models.}
\label{tab:LJ}
\smallskip
\begin{tabular}{@{}l@{~~}c@{~~}c@{~~}c@{~~}c@{~~}c@{~~}c@{~~}c@{~~}c@{~~}c@{~~}c@{~~}c@{~~}c@{}}
\toprule
& \multicolumn{4}{c}{PESQ ($\uparrow$)} & \multicolumn{4}{c}{STOI ($\uparrow$) } \\
\cmidrule(lr){2-5}
\cmidrule(lr){6-9}
Model $\setminus$ Iteration & 6    & 25    & 100   & 1000  & 6    & 25    & 100   & 1000 \\
\midrule 
WaveGrad \citep{chen_wavegrad_2020} & 2.78 & 3.194 & 3.211 & 3.290 & 0.924 & 0.957 & 0.958 & 0.959  \\
DDGM (ours)    & \textbf{3.07} & \textbf{3.208} & \textbf{3.214} & \textbf{3.308} & \textbf{0.948} & \textbf{0.972}  & \textbf{0.969} & \textbf{0.969} \\
\bottomrule
\end{tabular}
\smallskip
\centering
\caption{FID ($\downarrow$) score comparison for CelebA(64x64) dataset. Lower is better.}
\smallskip
\label{tab:celeba}
\begin{tabular}{lccccc}
\toprule
Model $\setminus$ Iteration & 10     & 20     & 50    & 100   & 1000      \\
\midrule 
DDPM \citep{ho2020denoising}   & 299.71 & 183.83 & 71.71 & 45.2  & \textbf{3.26}    \\
DDGM - Gamma Distribution DDPM (ours)   & \textbf{35.59}  & \textbf{28.24}  & \textbf{20.24} & \textbf{14.22} & 4.09   \\
\midrule 
DDIM  \citep{song_denoising_2020}   & 17.33  & 13.73  & 9.17  & 6.53  & 3.51      \\
DDGM - Gamma Distribution DDIM (ours)    & \textbf{11.64}  & \textbf{6.83}   & \textbf{4.28}   & \textbf{3.17}  & \textbf{2.92}      \\
\bottomrule
\end{tabular}
\smallskip
\centering
\caption{FID ($\downarrow$) score comparison for LSUN Church (256x256) dataset. Lower is better.}
\label{tab:church}
\smallskip
\begin{tabular}{lcccc}
\toprule
Model $\setminus$ Iteration & 10     & 20     & 50    & 100   \\ 
\midrule 
DDPM \citep{ho2020denoising}     &  51.56 &	23.37 &	11.16 &	8.27 \\ 
DDGM - Gamma Distribution DDPM (ours)      &  \textbf{28.56} &	\textbf{19.68} &	\textbf{10.53} &	\textbf{7.87} 
\\
\midrule 
DDIM  \citep{song_denoising_2020}     &  19.45 &	12.47 &	10.84	& 10.58 \\ 
DDGM - Gamma Distribution DDIM (ours)    &  \textbf{18.11} &	\textbf{11.32} &	\textbf{10.31} &	\textbf{8.75}  \\ 
\bottomrule
\end{tabular}
\end{table}

\subsection{Image Generation} 
Our model is based on the DDIM implementation available in \citep{DDIMgithub} (under the MIT license).  We trained our model on two image datasets (i) CelebA 64x64 \citep{liu2015faceattributes} and (ii) LSUN Church 256x256 \citep{yu15lsun}. The Fréchet Inception Distance (FID)~\citep{heusel2017gans} is used as the benchmark metric. For all experiments, similarly to previous work~\citep{song_denoising_2020}, we compute the FID score with $50,000$ generated images, using the torch-fidelity implementation~\citep{torchfidelity}. Similar to \citep{song_denoising_2020}, the training noise schedule $\beta_1, ... ,\beta_T$ is linear with values raging from $0.0001$ to $0.02$. For other hyperparameters (e.g. learning rate, batch size etc) we use the same parameters that appear in DDPM \citep{ho2020denoising}. We use eight Nvidia Volta V100 GPUs to train our models. The $\theta_0$ parameter for Gamma distribution set to $0.001$.

\noindent{\bf Results\quad}
We test our models with the inference procedure from DDPM \citep{ho2020denoising} and DDIM \citep{song_denoising_2020}. In Tab.~\ref{tab:celeba} we provide the FID score for CelebA (64x64) dataset \citep{liu2015faceattributes} (under non-commercial research purposes license). As can be seen for DDPM inference procedure for $10,20,50,100$ steps, the best results were obtained from the Gamma model, which improves results by a gap of $264$ FID scores for ten iterations. For $100$ iterations, the Gamma model improves results by $31$ FID scores. For $1000$ iterations, the best results were obtained from the DDPM model. Nevertheless, our Gamma model obtains results that are closer to the DDPM by a gap of $0.83$. For the DDIM procedure, the best results were obtained with the Gamma model for all number of iterations. Fig.~\ref{fig:evolution_celebA} presents samples generated by the three models. Our models provide better quality images when compared to DDPM and DDIM methods.

In Tab.~\ref{tab:church} we provide the FID score for the LSUN church dataset \citep{yu15lsun}.  As can be seen, the Gamma model improves results over the baseline for $10,20,50,100$ iterations.

\begin{figure}
\centering
\begin{tabular}{c}
\includegraphics[width=\linewidth]{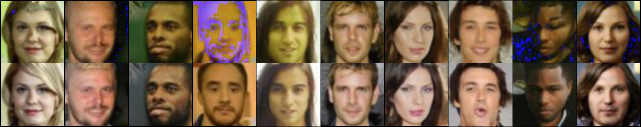}\\ 
\end{tabular}
\smallskip
\caption{Typical examples of images generated with $100$ iterations and $\eta=0$. For models trained with different noise distributions - (i) First row - Gaussian noise and (ii) Second row - Gamma noise. All models start from the same noise instance.}
\label{fig:evolution_celebA}
\end{figure}

\section{Conclusions} 
We present a novel Gamma diffusion model. The model employs a Gamma noise distribution. A key enabler for using these distributions is a closed-form formulation (Eq.~\ref{eq:close_form_single_gamma}) of the multi-step noising process, which allows for efficient training. We also present the reverse process and the variational lower bound for the Gamma diffusion model. The proposed model improves the quality of generated image and audio, as well as the speed of generation in comparison to conventional, Gaussian-based diffusion processes.

\subsubsection*{Acknowledgments}
This project has received funding from the European Research Council (ERC) under the European Unions Horizon 2020 research and innovation programme (grant ERC CoG 725974). The contribution of Eliya Nachmani is part of a Ph.D. thesis research conducted at Tel Aviv University.

\bibliography{iclr2022_conference}
\bibliographystyle{iclr2022_conference}

\appendix
\section{Proofs}

\subsection{Proof of lemma~\ref{lm:lm1}} 
\label{app:sec_1}
\primelemma*

\begin{proof}
    
    The first part of Eq.~\ref{eq:lemma2_1} is immediate. The variance part is also straightforward:
    \begin{equation*}
        V(g_t - E(g_t))= k_t{\theta_t}^2 = \beta_t
    \end{equation*}

    Eq.~\ref{eq:lemma2_2} is proved by induction on $t\in \{1, ...T\}$. 
    For $t=1$:
    \begin{equation*}
        x_1 = \sqrt{1-\beta_1}x_0 + g_1 - E(g_1)
    \end{equation*}

    since $\bar k_1 = k_1$, $\bar g_1 = g_1$. We also have that $\sqrt{1-\beta_1} = \sqrt{\bar \alpha_1}$. Thus we have:

\begin{equation*}
    x_1 = \sqrt{\bar \alpha_1}x_0 + (\bar g_1 - E(\bar g_1))
\end{equation*}

Assume Eq.~\ref{eq:lemma2_2} holds for some $t\in \{1, ...T\}$. The next iteration is obtained as
\begin{align}
    x_{t+1} &= \sqrt{1- \beta_{t+1}} x_t + g_{t+1} - E(g_{t+1})\\
           & = \sqrt{1- \beta_{t+1}} (\sqrt{\bar \alpha_t}x_0 + (\bar g_t - E(\bar g_t))) + g_{t+1} - E(g_{t+1}) \\
           & = \sqrt{\bar \alpha_{t+1}}x_0 + \sqrt{1- \beta_{t+1}} \bar g_t + g_{t+1} - (\sqrt{1- \beta_{t+1}}E(\bar g_t) + E(g_{t+1}))
\end{align}

It remains to be proven that (i) $\sqrt{1- \beta_{t+1}} \bar g_t + g_{t+1} = \bar g_{t+1}$ and (ii) $\sqrt{1- \beta_{t+1}}E(\bar g_t) + E(g_{t+1}) = E(\bar g_{t+1})$. Since $\bar g_t \sim \Gamma(\bar k_t, \theta_t)$ hold, then:
\begin{align*}
    \sqrt{1- \beta_{t+1}} \bar g_t & \sim \Gamma(\bar k_t, \sqrt{1- \beta_{t+1}}\theta_t) = \Gamma(\bar k_t, \theta_{t+1}) 
\end{align*}
Therefore, we prove (i):
\begin{align*}
    \sqrt{1- \beta_{t+1}} \bar g_t + g_{t+1} &\sim \Gamma(\bar k_t+ k_{t+1}, \theta_{t+1}) = \Gamma(\bar k_{t+1}, \theta_{t+1})
\end{align*}
which implies that $\sqrt{1- \beta_{t+1}} \bar g_t + g_{t+1}$ and $\bar g_{t+1}$ have the same probability distribution.

Furthermore, by the linearity of the expectation, one can obtain (ii):
\begin{align*}
     \sqrt{1- \beta_{t+1}}E(\bar g_t) + E(g_{t+1}) &=  E(\sqrt{1- \beta_{t+1}} \bar g_t + g_{t+1}) \\
    &= E(\bar g_{t+1})
\end{align*}

Thus, we have:
\begin{equation*}
    x_{t+1} = \sqrt{\bar \alpha_{t+1}}x_0 + (\bar g_{t+1} - E(\bar g_{t+1}))
\end{equation*}
which ends the proof by induction.
\end{proof}

\subsection{Proof of lemma~\ref{lm:lm2}} 
\label{app:sec_2}
\primelemmasec*
\begin{proof}
From Eq.\ref{eq:vlb_end}, the variational lower bound of DDGM is given by $L_{t-1} \leq \left( C_1 +C_2 + \frac{C_3}{\bar g_t} + \frac{C_4}{\bar g_{t-1}}\right) |x_0 - \hat x_0|$. Substitute Eq.\ref{eq:hatx_0} and Eq.\ref{eq:close_form_single_gamma} to the variational lower bound we have:
\begin{align}
    L_{t-1} &\leq \left( C_1 +C_2 + \frac{C_3}{\bar g_t} + \frac{C_4}{\bar g_{t-1}}\right) \left |x_0 - \hat x_0\right| \\
    &= \left( C_1 +C_2 + \frac{C_3}{\bar g_t} + \frac{C_4}{\bar g_{t-1}}\right) \left | x_0 - \dfrac{x_t -\sqrt{1 - \bar \alpha_t} \varepsilon_\theta(x_t, t)}{\sqrt{\bar \alpha_t}}\right| \\
    &= \left( C_1 +C_2 + \frac{C_3}{\bar g_t} + \frac{C_4}{\bar g_{t-1}}\right) \frac{1}{\sqrt{\bar \alpha_t}} \left | \sqrt{\bar \alpha_t}x_0 - x_t +\sqrt{1 - \bar \alpha_t} \varepsilon_\theta(x_t, t)\right| \\
    &= \left( C_1 +C_2 + \frac{C_3}{\bar g_t} + \frac{C_4}{\bar g_{t-1}}\right) \frac{1}{\sqrt{\bar \alpha_t}} \left | \sqrt{\bar \alpha_t}x_0 - \sqrt{\bar \alpha_t} x_0 - (\bar g_t - \bar k_t\theta_t) +\sqrt{1 - \bar \alpha_t} \varepsilon_\theta(x_t, t) \right| \\
    &= \left( C_1 +C_2 + \frac{C_3}{\bar g_t} + \frac{C_4}{\bar g_{t-1}}\right) \frac{1}{\sqrt{\bar \alpha_t}}  \left |(\bar g_t - \bar k_t\theta_t) - \sqrt{1 - \bar \alpha_t} \varepsilon_\theta(x_t, t) \right| \\
    &= \left( C_1 +C_2 + \frac{C_3}{\bar g_t} + \frac{C_4}{\bar g_{t-1}}\right) \frac{\sqrt{1 - \bar \alpha_t}}{\sqrt{\bar \alpha_t}} \left |\frac{\bar g_t - \bar k_t\theta_t}{\sqrt{1 - \bar \alpha_t}} - \varepsilon_\theta(x_t, t) \right|
\end{align}

Since we are minimizing the variational lower bound, one can drop the constant term $\left( C_1 +C_2 + \frac{C_3}{\bar g_t} + \frac{C_4}{\bar g_{t-1}}\right) \frac{\sqrt{1 - \bar \alpha_t}}{\sqrt{\bar \alpha_t}}$. Therefore, minimizing the variational lower bound is equal to minimizing the term $\left |\frac{\bar g_t - \bar k_t\theta_t}{\sqrt{1 - \bar \alpha_t}} - \varepsilon_\theta(x_t, t) \right|$.
\end{proof}
\end{document}